\documentclass[english,oneside,12pt]{article}

\usepackage[utf8]{inputenc}
 
\usepackage[numbers]{natbib}
\usepackage{graphicx}
\usepackage{amsmath}
\usepackage{amsfonts}
\usepackage{amssymb}
\usepackage{amsthm}
\usepackage{booktabs}
\usepackage{multirow}
\usepackage{makecell}
\usepackage[super]{nth}
\usepackage{longtable}
\usepackage[dvipsnames,table]{xcolor}
\usepackage{tikz}
\usepackage{pgfplots}
\usepackage{units}
\usepackage[height=20cm,width=16cm]{geometry}

\usepackage[linesnumbered,vlined,ruled]{algorithm2e}
\usepackage{subcaption}
\usepackage{thm-restate}

\usepackage{hyperref} 

\newcommand\Sec[1] {Sec.~\ref{#1}}
\newcommand\Fig[1] {Fig.~\ref{#1}}
\newcommand\Table[1] {Tbl.~\ref{#1}}

\newcommand{\qm}[1]{``#1''}

\newcommand{\call}{\mathcal{L}}

\newcommand{\calx}{\ensuremath{\mathcal{X}}}
\newcommand{\caly}{\ensuremath{\mathcal{Y}}}

\newcommand{\vbayes}{V_{1}}
\newcommand{\bayesleak}{\call_{1}}

\newcommand\Joint[2]		{#1{\kern-0.6pt\smalltriangleright\kern.3pt}#2} 

\DeclareFontFamily{U}{mathb}{\hyphenchar\font45}
\DeclareFontShape{U}{mathb}{m}{n}{
	<5> <6> <7> <8> <9> <10> gen * mathb
	<10.95> mathb10 <12> <14.4> <17.28> <20.74> <24.88> mathb12
}{}
\DeclareSymbolFont{mathb}{U}{mathb}{m}{n}
\DeclareMathSymbol{\smalltriangleright}{2}{mathb}{"9B}

\LTcapwidth=\textwidth

\newcolumntype{L}[1]{>{\raggedright\let\newline\\\arraybackslash\hspace{0pt}}m{#1}}
\newcolumntype{C}[1]{>{\centering\let\newline\\\arraybackslash\hspace{0pt}}m{#1}}
\newcolumntype{R}[1]{>{\raggedleft\let\newline\\\arraybackslash\hspace{0pt}}m{#1}}

\newcommand{\dataset}[1] {\textsc{#1}}

\newcommand{\tabimporttrans}{\textsc{De-identified transactions}\xspace}
\newcommand{\tabsumbycity}{\textsc{Summary by city}\xspace}
\newcommand{\tabimporters}{\textsc{Importers}\xspace}
\newcommand{\tabsumbyncm}{\textsc{Summary by state}\xspace}

\newcommand{\rfb}{Tax and Customs Administration\xspace}

\rowcolors{1}{white}{gray!10}

\begin{document}

\title{\textsc{A novel reconstruction attack on foreign-trade official statistics, with a Brazilian case study}}
\author{Danilo Fabrino Favato\footnote{dfavato@dcc.ufmg.br --- remaining affiliations in the end of the manuscript.} \and Gabriel Coutinho \and M\'{a}rio S.\ Alvim \and Natasha Fernandes}
\date{\today}
\maketitle
\vspace{-0.8cm}

\begin{abstract}{
    In this paper we describe, formalize, implement, and experimentally evaluate a novel transaction re-identification attack against official foreign-trade statistics releases in Brazil. 
    The attack's goal is to re-identify the importers of foreign-trade transactions (by revealing the identity of the company performing that transaction), which consequently violates those importers’ fiscal secrecy (by revealing sensitive information: the value and volume of traded goods). 
    We provide a mathematical formalization of this fiscal secrecy problem using principles from  the framework of quantitative information flow (QIF), then carefully identify the main sources of imprecision in the official data releases used as auxiliary information in the attack, and model transaction re-construction as a linear optimization problem solvable through integer linear programming (ILP). We show that this problem is NP-complete, and provide a methodology to identify tractable instances. We exemplify the feasibility of our attack by performing 2,003 transaction re-identifications that in total amount to more than \$137M, and affect 348 Brazilian companies.
    Further, since similar statistics are produced by other statistical agencies, our attack is of broader concern.}
\end{abstract}

\begin{center}
	\textbf{Keywords} \\ 
	\textsc{Quantitative information flow, Integer optimization programming, Database reconstruction attack, Foreign trade statistics}
\end{center}
  
  \vspace{1cm}

\maketitle

\section{Introduction}\label{sec:intro}

\textit{Foreign-trade statistics} comprise national import and export summaries, and are instrumental both for 
national governments and for official supranational 
and international organizations.
These statistics are used, for example, to produce a country's National Accounts, which is the official bookkeeping of its economic activity, 
as well as to provide accurate measures of that nation's 
Gross Domestic Product (GDP) \cite{sna2009}.

The \textit{utility} of accurate foreign-trade statistics, however,
surpasses governments, as they can provide valuable information
to society as a whole. 
The release to the public of accurate trade statistics enables many useful analyses including, for example,
 the identification of harmful economic practices such as forgery and predatory prices (dumping)~\cite{a_vi_gatt_1994},
    market and economic studies for new businesses planning, and
    assessment of the impact of governmental policies.

However, utility to society is not the government's only concern when releasing foreign-trade statistics.
In order to foster a fair, free-market environment,
several countries grant companies legal rights to some level of
\textit{fiscal secrecy} about their foreign-trade transactions.

In Brazil, federal legislation requires public disclosure of foreign-trade statistics,\footnote{Item X of Article 92, Decree 9745 of April \nth{8}, 2019 and Article 1, Item XII of the Brazilian Tax and Customs Administration internal regiment (Annex I Ordinance 284 of July \nth{27}, 2020).} but leaves the specific granularity of what is disclosed to each governmental agency's discretion.
On the other hand, Brazilian fiscal secrecy legislation recognizes as 
sensitive information ---and therefore worthy of protection---
the values and volumes of goods traded by individual companies.
It also forbids governmental entities from publishing this kind 
of information without the usage of disclosure control techniques protecting the identity of the importer in each 
transaction.\footnote{Article 198, National Tax Law 5.172, October \nth{25}, 1966.}
In order to balance the utility (to the general public) 
of foreign-trade statistics and the 
legally required fiscal secrecy (of companies),
 Brazilian governmental agencies have (since 2007) been 
publishing
datasets of foreign-trade transactions anonymized using techniques such as de-identification, generalization and suppression\footnote{Ordinance SRF 306, March \nth{22} of 2007.} (detailed in Section~\ref{sec:overview-datasets}  ahead).

In this paper we describe, formalize, implement, and experimentally evaluate a novel transaction re-identification attack against official 
foreign-trade statistics releases in Brazil. 
The attack's goal is to re-identify the importers of foreign-trade transactions (by revealing the the identity of the company performing that transaction), which consequently violates those importers’
fiscal secrecy (by revealing sensitive information: the
value and volume of traded goods). 
After properly formalizing and implementing our attack, we demonstrate its effectiveness against Brazilian releases of official foreign-trade statistics by performing 2,003 transaction re-identifications that in total amount to over  \$100M in value, and affects 348 companies.

It is noteworthy that, although in the past the 
methodology for collection and organization of foreign-trade statistics 
varied significantly from country to country, nowadays a
standardization provided by the United Nations (UN)~\cite{un_imts_2010} is adopted by many countries, including Brazil.
This means that Brazilian foreign-trade statistics analyzed in this paper are highly compatible with those of other countries, and therefore our
novel attack might also be of concern to those countries.

\subsection{Contributions and plan of the paper}

The main contributions of this paper are the following.

\begin{enumerate}
    \item We identify  and describe a novel transaction re-identification attack that violates the fiscal secrecy of companies by reconstructing foreign-trade transaction data
    from publicly available datasets in the context of Brazil.
    Since production of Brazilian foreign-trade statistics follows the UN's standard methodology, our attack may be of concern to other countries as well.
    (\Sec{sec:overview}).
    
    \item We provide a mathematical formalization of the fiscal secrecy problem of transaction re-identification using principles from  the framework of \emph{quantitative information flow} (QIF)~\cite{ScienceOfQIF},
    which has sound information- and decision-theoretic grounds to naturally model an adversary's knowledge, goals, and capabilities, as well as the leakage of information caused by an attack. (\Sec{sec:model}).
    
    \item We carefully identify the main sources of imprecision in the 
    official foreign-trade statistics used as side-information in our attack, including:
    \begin{enumerate}
        \item
        the splitting of transaction information across various datasets;
        \item
        the sanitization techniques employed by the Brazilian government in the publishing of these statistics, such as de-identification, generalization, and suppression; and
        \item
        rounding errors introduced (voluntarily or not) in the data publishing process.
    \end{enumerate}
    To cope with such sources of uncertainty, we
    model transaction re-construction as a linear optimization problem solvable through integer linear programming (ILP), and show that it is NP-complete. We then provide a methodology to identify tractable instances of the problem, as well as to simplify some intractable instances by finding further constraints in order to try to make them tractable (\Sec{sec:implementation}).
    
    \item We exemplify the feasibility of our attack and its algorithmic solution through a concrete case study in which over 2,000 transactions are re-identified and the fiscal secrecy of hundreds of companies are violated in Brazil, using as computational resources only a modest personal computer (\Sec{sec:experiments}).
\end{enumerate}

Finally, in \Sec{sec:related-work} we discuss related work, and in \Sec{sec:conclusion} we present our conclusions and future prospects.

\subsection{Ethical Considerations}
\label{sec:ethical_cons}

The results of this study have already been presented to the relevant government agencies ---the Brazilian Customs and Tax administration---
in order to give them time to act. On December \nth{16} 2021 they removed public access to 
 one of their main databases, which is sufficient to prevent the type of attack presented here. For ethical reasons we have also chosen not to disclose in this paper other points of access to the data used in the description of our attack. Throughout this text we omit the underlying business direct identifiers and, in particular, we do not provide identifiers for businesses whose privacy is affected by this attack.
 
On May \nth{4} we contacted the UN Statistical Commission and the Tax and Customs administration of a particular country ---whose name we omit here for precaution\footnote{This country was directly contacted as it is, to our knowledge, the one whose disclosure methodology most closely resembles the Brazilian, and thus for which our attack could be easily adapted.}. Neither had replied to our e-mail by the time we finished this document.

\section{The problem, the datasets, and the attack}\label{sec:overview}

In this section we provide an overview of the problem tackled in this paper and motivate why it is a serious fiscal secrecy matter. 
We start with a description of the Brazilian environment for production, 
processing, and publication of foreign-trade information.
We then describe the publicly available data used in our attack, 
and provide a concrete attack instance used to compromise fiscal secrecy. 
Since Brazilian foreign-trade statistics methodology is highly compatible
with that of other countries, the attack we identify may be replicable in
those countries as well.
A proper formalization of the attack ---including its general form--- and
full algorithmic solutions are deferred to Sections~\ref{sec:model} and~\ref{sec:implementation}, respectively.

\subsection{Fiscal secrecy vs.\ societal utility in foreign-trade statistics}\label{overview-trade}

The current legal framework for the collection, treatment, usage,
and publication of individuals' and companies' data by the government
has two conflicting goals. 
On one hand, there is an increasing desire for 
\textit{Open Government Data}, which ``promotes transparency, accountability and value creation by making government data available to all''.\footnote{\url{https://web.archive.org/web/20201120073847/https://www.oecd.org/gov/digital-government/open-government-data.htm}} 
This transparency is extremely useful for civil society to, e. g., follow and evaluate the effectiveness of public policies, or detect
fraud in the use of public money, or even malpractices by ill-intended business.
In this context, it is desirable that the published data have high \textit{utility}, in the sense of being as accurate, detailed, complete, and freely available as possible. 

On the other hand, there are various concerns about the unauthorized disclosure of sensitive information about individuals or companies, a worry that has sprouted many data protection laws around the globe.\footnote{Notable examples or privacy legislation include the General Data Protection Regulation in the European Union (GDPR 2016/679), the \textit{Lei Geral de Proteção de Dados Pessoais} in Brazil (LGPD 13.709/2018), the Personal Information Protection and Electronic Documents Act in Canada (PIPEDA April \nth{13} 2000), and The Data Privacy Act in the Philippines (2012).}
Although these novel protection laws, as well as media and public attention, tend to focus on protecting data about individuals (\textit{privacy}), in this paper we focus on the unwanted disclosure of information about commercial transactions between companies (\textit{fiscal secrecy}). 
A recent example of the risks associated with the latter is the case in which the American National Security Agency (NSA) was accused of worldwide industrial espionage and the shock produced by these allegations.\footnote{\url{https://web.archive.org/web/20210417032720if\_/https://www.reuters.com/article/us-security-snowden-germany-idUSBREA0P0DE20140126}} 

For a more concrete example of the importance of fiscal secrecy,
consider a small company $A$ that, through its own research and merit, 
has established an agreement with a new foreign supplier which has 
given $A$ a significant edge on the market.
Now say that $B$, a very large (and ill-intentioned) competitor of $A$, 
is interested in finding out $A$'s new supplier to
approach it and either offer it an exclusivity contract (in which the supplier could only sell to $B$), or just start buying all the supplier's capacity in order to starve $A$. 
This could lead to $A$ being pushed out of the market and to significant
losses for its owners and employees.

In general, violations of fiscal secrecy can be harmful to the idea of free markets.
Indeed, \textit{symmetric information} is recognized as one the key requirements for a free competitive market~\cite{LofgrenKarl-Gustaf2002MwAI}.
Of course, in real world markets there is no such a thing as
perfectly symmetric information. 
Each player will know something different about their market, and that is considered acceptable when such information was acquired through legal methods such as research, qualified personnel, etc. Nonetheless, when the information is acquired illegally it is usually tagged as ``privileged'' or \qm{espionage}, and seen as detrimental in fair free-market competition.

In this context, fiscal secrecy violations ---whereby some companies may 
illegally obtain secret information about other companies--- are of concrete concern.
It is thus a challenge for governments to balance the utility (to the general public) of foreign-trade statistics and the fiscal secrecy (of companies).

\subsection{The data release methodology}\label{sec:overview-datasets}

\Fig{fig:reidentification} provides an overview of 
the process of production, compilation, sanitization, and publication of  foreign-trade statistics in Brazil.

\begin{figure}[h]
	\begin{center}
	\includegraphics[width=.5\textwidth]{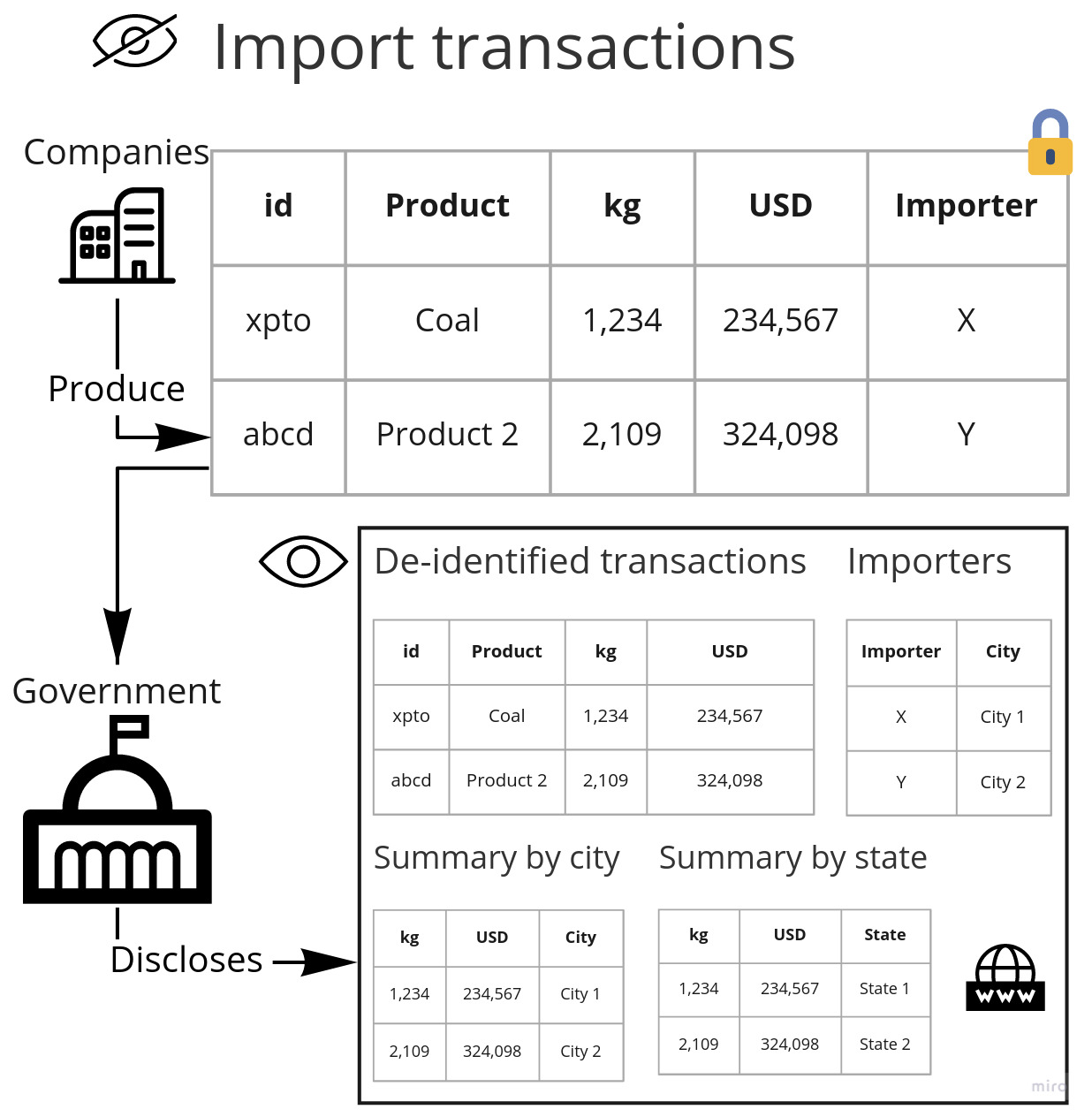}
	\end{center}
	\caption{Overview of the generation of foreign-trade statistics.}\label{fig:reidentification}
\end{figure}

The process begins with each company in the country
producing and submitting
to Siscomex (the Brazilian International Commerce System)\footnote{\url{https://www.gov.br/receitafederal/pt-br/assuntos/aduana-e-comercio-exterior/manuais/despacho-de-importacao/topicos-1/conceitos-e-definicoes/despacho-de-importacao}}
a legally required \emph{Declaration of Imports} (DI) 
whenever an import transaction occurs.

The DI contains detailed information about each transaction
performed by the company, including the goods imported, its quantity, from whom, and at what price and date. It is a document analogous to an invoice in the foreign-trade context, 
and it clearly states the seller and the importer of each transaction.

Upon receiving all DIs from all over the country, 
the government compiles them into a single collection of 
foreign-trade transactions.
This raw-data collection is only used internally by the 
government, but a sanitized form of it is published after the
application of the following sanitization techniques:
\emph{de-identification}, by which obvious individual identifiers (such as an importer's name) are removed; 
\emph{generalization}, by which data is provided in coarser granularity (e. g. when transaction values are grouped by city); and
\emph{suppression}, in which part of the data is removed from the collection. 
The goal of these disclosure methods is to break the connection between the goods imported and the companies behind such transactions.

The end result of this process are the 4 datasets described below, 
which have been publicly available in the period we conducted this study, 
and were updated monthly.
For brevity, only relevant fields from each dataset are included here; full dataset details are provided in Appendix~\ref{sec:sargasso-dataset-structure}. In addition, as explained in \Sec{sec:ethical_cons}, we have chosen not to disclose public access urls for these datasets.

\begin{enumerate}
\item The \dataset{De-identified transactions} dataset (Tbl.~\ref{tab:rfb-fields}) contains detailed records
from the original compilation of all companies' 
DIs, with all importer's direct identifiers omitted.\footnote{During our study,  and until December \nth{16} 2021,
	the \rfb kept an up-to-date version of the dataset on its website.}
	
	\begin{table}[!htbp]
	    \centering
	    \begin{subtable}[b]{\linewidth}
	        \centering
	        \begin{tabular}{L{0.28\linewidth}L{0.62\linewidth}}
		    \textbf{Field} & \textbf{Description} 
		    \\\hline
		\texttt{NUMERO DE ORDEM} & Transaction identification number.
		\\ 
		\texttt{ANOMES} & Year and month of transaction registry. 
		\\ 
		\texttt{COD.NCM} & DI header: NCM of imported goods.
		\\ 
		\texttt{PAIS DE ORIGEM} & DI header: country of origin.
		\\ 
		\texttt{PESO LIQUIDO} & DI header: net weight in kg.
		\\ 
		\texttt{VMLE DOLAR} & DI header: FOB value in USD.
		\\
		\texttt{NAT. INFORMACAO} & DI header: nature of the operation. 
		\\	
	    \end{tabular}
	    \caption{Relevant fields from the \dataset{De-identified transactions} dataset.}
	        \label{tab:rfb-fields} 
        \end{subtable}
	    
	   	\begin{subtable}[b]{\linewidth}
	        \centering
	        \begin{tabular}{L{0.28\linewidth}L{0.62\linewidth}
	        }
		\textbf{Field} & \textbf{Description} 
		\\\hline
		\texttt{CNPJ}		& Company's registration number.		
		\\
		\texttt{EMPRESA}		& Company's name.
		\\
		\texttt{MUNICÍPIO}	& Company's city.
		\\
		\texttt{UF}			& Company's federation state.
		\\			
	\end{tabular}
	\caption{Relevant fields from the \dataset{Importers} dataset.}
	\label{tab:importers-fields}
    \end{subtable}	
	    
	 \begin{subtable}[b]{\linewidth}
	        \centering
	        \begin{tabular}{L{0.28\linewidth}L{0.62\linewidth}}
			\textbf{Field} & \textbf{Description}
			\\\hline
			\texttt{CO\_ANO}			& Year of registry.
			\\
			\texttt{CO\_MES}			& Month of registry.
			\\
			\texttt{SH4}				& Imported goods' SH4 
			code.		
			\\
			\texttt{CO\_PAIS}		& Imported goods' country
			of origin.
			\\
			\texttt{SG\_UF\_MUN}		& Importer's federation
			state.
			\\
			\texttt{CO\_MUN}			& Importer's city.
			\\
			\texttt{KG\_LIQUIDO}		& Imported goods' total
			net weight in kg.
			\\
			\texttt{VL\_FOB}			& Imported goods' total
			value in USD.
			\\			
		\end{tabular}
		\caption{Relevant fields from the \dataset{Summary by city} dataset.}
		\label{tab:hs4-city-fields} 
		\end{subtable}
		
		\begin{subtable}[b]{\linewidth}
	        \centering
	        \begin{tabular}{L{0.28\linewidth}L{0.62\linewidth}}
		\textbf{Field} & \textbf{Description}
		\\\hline
		\texttt{CO\_ANO}			& Year of registry.
		\\
		\texttt{CO\_MES}			& Month of registry.
		\\
		\texttt{CO\_NCM}			& Imported goods NCM
		code number.
		\\
		\texttt{CO\_PAIS}		& Imported goods' country
		of origin.
		\\
		\texttt{SG\_UF\_NCM}	& Importer' federation state.
		\\
		\texttt{KG\_LIQUIDO}		& Imported goods' total net weight in kg.
		\\
		\texttt{VL\_FOB}			& Imported goods' total value in USD.
		\\
	\end{tabular}
	\caption{Relevant fields from the \tabsumbyncm dataset.}
	\label{tab:ncm-uf-fields} 
\end{subtable}
		
		\caption{Relevant fields of published datasets, whose full descriptions can be found in Appendix~\ref{sec:sargasso-dataset-structure}.}
		\label{tab:published-datasets}
			
	\end{table}
	
	\item The \dataset{Importers} dataset (Tbl.~\ref{tab:importers-fields}) contains data about businesses that have conducted at least one import operation in the current year. New businesses were added to the list each month, so
	anyone with access to the dataset could pinpoint the exact month in which a given importer started importing by tracking the additions made to the list.

\item The \dataset{Summary by city} dataset (Tbl.~\ref{tab:hs4-city-fields}) 
contains data of transactions' value and weight aggregated by city.
	
	\item The \tabsumbyncm dataset,
	(Tbl.~\ref{tab:ncm-uf-fields})
	contains aggregated data for value, weight, freight, insurance and quantity of transactions, grouped by states.
	This dataset contains the official figures of foreign commerce 
	for the National Accounting System, which are used to calculate the Brazilian GDP~\cite{NotaTecnicaSecex}. 
	
\end{enumerate}    

\subsection{Overview of the attack, and the correlations it exploits}

The crucial challenge in our transaction re-identification attack is how to reconstruct a row of the original raw-data dataset from the publicly released, sanitized datasets described in the previous section.
Clearly, from the reconstructed transaction row the adversary immediately 
recovers the transaction's importer and its sensitive values.
The attack exploits the following correlations among the datasets published by the government.
    
    \begin{itemize}
        \item  \textbf{NCM fields.}
    The fields \texttt{COD.NCM} in \tabimporttrans, \texttt{SH4} in \tabsumbycity and \texttt{CO\_NCM} in \tabsumbyncm are all related to the \emph{Common Mercosur Nomenclature} (NCM), which is a regional product categorization system used in the Mercosur Economic Region since 1995.\footnote{\url{https://receita.economia.gov.br/orientacao/aduaneira/classificacao-fiscal-de-mercadorias/ncm}} 
    It was derived from (and is highly compatible with) the \emph{Harmonized Commodity Description and Coding System} (HS), which was designed and is maintained by the World Customs Organization.\footnote{\url{http://www.wcoomd.org/en/topics/nomenclature/overview/what-is-the-harmonized-system.aspx}} NCM codes are 8-digits long and follow a hierarchical organization. Eg., the NCM code 08051000 identifies \textit{Fresh oranges}. Its first four digits, 0805, represent the hierarchical level labeled SH4 
    \textit{Citrus fruit, fresh or dried}. Whenever a NCM code appears in \tabimporttrans, it also appears in \tabsumbyncm, and its SH4 appears in \tabsumbycity.
    
    \item   \textbf{Final destination fields.}
    The fields \texttt{MUNICÍPIO} in \tabimporters, \texttt{CO\_MUN} in \tabsumbycity and \texttt{CO\_URF} in \tabsumbyncm are all related to the importer's official address. The first two fields refer to the importer's city, and the last refers to the importer's state. Whenever a city appears in \tabsumbycity, it also appears in \tabimporters, and its state appears in \tabsumbyncm.
    
    \item \textbf{Country of origin fields.}
    The fields \texttt{PAIS DE ORIGEM} in \tabimporttrans and \texttt{CO\_PAIS} in both \tabsumbycity and \tabsumbyncm are all related to the country of origin of the imported goods. If a country appears in \tabimporttrans, it also appears in \tabsumbycity and in \tabsumbyncm.
    \end{itemize}

In the absence of any anonymization technique, the
solution to the problem of reconstructing transaction rows from the
public datasets could be obtained via a system 
of linear equations, whose variables would be the quantities 
(values and weights) in each transaction, whose constraints would be the summary statistics, and whose solution space would be defined by the whole set of transactions. 
Even in this simpler situation, however, it would still be possible that
multiple solutions existed, ie., multiple mappings of transactions to importers that satisfied all constraints. 
But since anonymization techniques such as suppression and generalization are employed on the published data, further uncertainty
is added to the problem, increasing the number of possible solutions.

For the reasons above, we model transaction reconstruction as a 
linear optimization problem solvable through integer linear programming.
In Sec.~\ref{sec:model} we formalize the attack,
and in Sec.~\ref{sec:implementation} we detail its algorithmic
solution, including 
how we cope with further uncertainties introduced by imperfect data.
Before that, however, we present a concrete motivating example of
our attack.

\subsection{A motivating example}\label{sec:overview-example}

Here we provide a concrete instance of
our transaction re-identification attack.
which is formalized in Sec.~\ref{sec:model}.

In this attack, depicted in Fig.~\ref{fig:attack-flow},
an adversary with very limited prior knowledge applies 
integer programming techniques (the \texttt{ATTACK ALGORITHM} detailed in Sec.~\ref{sec:implementation}) to publicly available 
data.
Her goal is to reconstruct the link between a transaction chosen
as target in 
the \dataset{De-identified transactions} dataset and a company in the \dataset{Importers} dataset, effectively re-identifying the target transaction's importer and inferring its sensitive attributes (value and volume).

In our concrete example we describe how the adversary
is able to re-identify the importer of a real transaction 
valued at $\sim${\$3400} USD, performed in January 2021. For ethical reasons, we omit the importer's identity.

\begin{figure*}[ht]
    \centering
    \includegraphics[width=\textwidth]{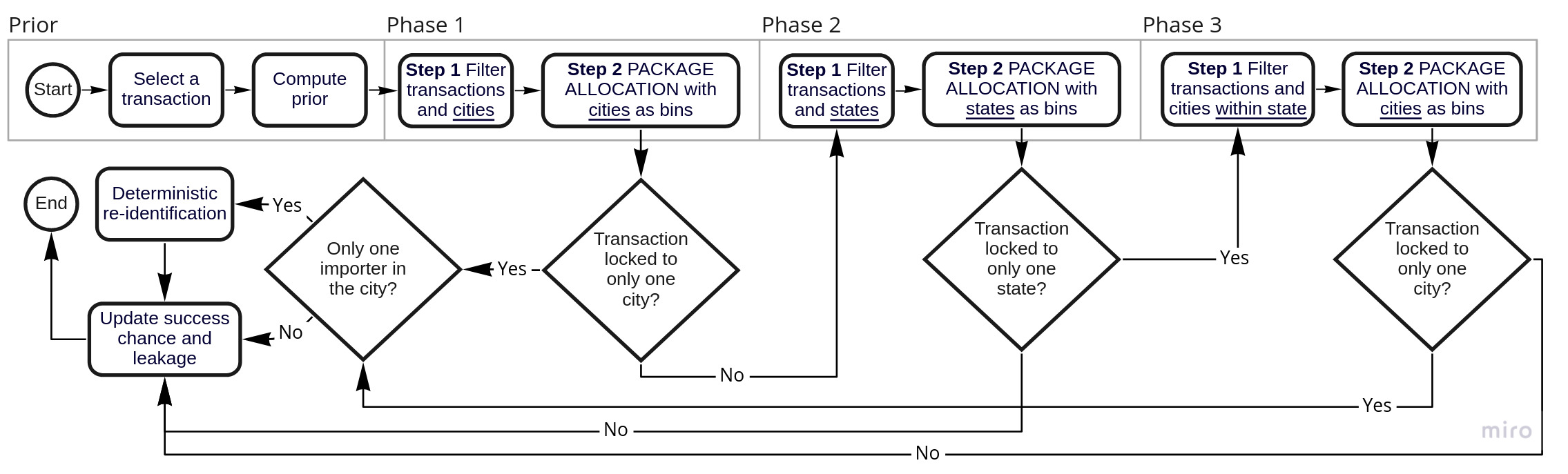}
    \caption{Phases and steps of the attack algorithm.}
    \label{fig:attack-flow}
\end{figure*}

\textbf{Before the attack.}
The adversary starts with access to only the publicly available \dataset{De-identified transactions} and \dataset{Importers} datasets (Tbls.~\ref{tab:rfb-fields} and~\ref{tab:importers-fields}, respectively).
She picks as her target for re-identification the 
transaction  
specified in Tbl.~\ref{tab:example-specific-target}, which 
we refer to by its identifying \texttt{NR ORDEM}: 10653400001. 

\begin{table}[!htbp]
		\centering
			\begin{tabular}{ll}
				\textbf{NR ORDEM} & 10653400001 \\ 
				\textbf{Description} &
                Cotton T-400 [...] Metros:1.013,70(09 rolos) \\
				\textbf{Quantity} & 420.69 kg \\ 
				\textbf{Value (USD)} & 3,388.41 \\
				\textbf{NCM} & 52083900 \\
				\textbf{Origin} & China \\
			\end{tabular}
		\caption{Target transaction for re-identification.}
		\label{tab:example-specific-target}
	\end{table}

Although real-world adversaries typically have some common knowledge
about foreign trade (eg., they know that a clothing manufacturer is more likely to buy 400 kg of cotton cloth than a toy store is),
here we are conservative and assume a very modest 
adversary who lacks any such common knowledge.
She just assumes, 
a priori, that every company is equally likely to be the 
importer of any given transaction.
This means that our adversary's best strategy to re-identify 
the target transaction before the attack is performed 
is just to blindly guess which company among the ${18,430}$ 
present in the \dataset{Importers} dataset is the actual importer. Therefore, the adversary's probability of a successful 
transaction re-identification at this point is $\nicefrac{1}{18,430} \approx 0.005\%$.

\textbf{Attack execution.}
During the attack the adversary gains access to the \dataset{Summary by city} and \dataset{Summary by state}
datasets, and uses them as auxiliary information to re-identify the target transaction's importer. 
She breaks her task into two parts as follows.
\begin{itemize}
    \item First, she tries to infer the city of the target transaction.
    This is done by using \tabsumbycity in Phase 1 
    to try to link the target transaction directly to a city.
    If that direct step is computationally intractable (as
    we shall see in \Sec{sec:implementation}, the \texttt{PACKAGE ALLOCATION} algorithm used to link datasets is  exponential), the adversary tries to find further constraints to simplify the problem.
    For that, she uses \tabsumbycity in conjunction with \tabsumbyncm in Phase 2 to try to first link the target transaction to a state, and only then she moves to Phase 3 and tries to link the target transaction to a city, given that the state is already known.
    \item Second, if the adversary has successfully linked the target transaction to only one city, she guesses as the importer a company based on that city, found by querying the \tabimporters dataset.
\end{itemize}

We now describe the attack phases in more detail.

\textbf{Attack Phase 1: Trying to link the target transaction directly to a city.}
To try to directly link the target transaction (from \Table{tab:example-specific-target}) to its city
of occurrence, the adversary will perform two steps:
\begin{enumerate}
    \item Filter \tabsumbycity and \tabimporttrans by SH4 = 5208 and country of origin = China (same parameters as target transaction). This yields 19 cities and 99 transactions.
    \item Try to run the \texttt{PACKAGE ALLOCATION} algorithm to link the target transaction to one of the cities.
\end{enumerate}
Step 1 already reduces the number of possible importers from over 18,000 to fewer than 5,000 (just the ones linked to those 19 cities). However, for our computationally bounded adversary, Step 2 is intractable. Therefore she proceeds to Phase 2.

\textbf{Attack Phase 2: Reducing the instance's size by first trying to link the target transaction to a state.}
Here the adversary tries to link the target transaction to its state as an intermediate step for finding the transaction's city. This is also done in two steps:
\begin{enumerate}
    \item Filter \tabsumbyncm and \tabimporttrans by NCM = 52083900 and country of origin = China (same parameters as target transaction). The results of these queries are show in Tbls.~\ref{tab:example-possible-states} and \ref{tab:example-all-transactions} respectively.
    \item Try to run the \texttt{PACKAGE ALLOCATION} algorithm to link the target transaction to one of the states.
\end{enumerate}

The key idea behind this linkage algorithm is that the adversary knows that there must exist at least one valid partition of these transactions among states that produces the same sums observed in \Table{tab:example-possible-states}. 
(Notice that the totals on both tables are the same, ignoring the decimal places; we shall discuss rounding errors in \Sec{sec:implementation-rounding}).

\begin{table}[!htbp]
\centering
        
	\begin{subtable}[b]{\linewidth}
		\centering
		\begin{tabular}{rrr}
				\textbf{NR ORDEM} 	& \textbf{Value USD} 		& \textbf{kg} \\ \hline
				\underline{10653400001}		&3,388.41			&420.69		\\
				08986300003                 & 111.91            & 23.1 \\
                02113100001                 & 1,116.52          & 242.86 \\
                05566500003                 & 19,856.21         & 4,055 \\
                12343400002                 & 20,216.8          & 2,091.3 \\
                12319800004                 & 1,346.64          & 136 \\
                12634700012                 & 7,698.19          & 918.25 \\
                02898100002                 & 57,194.45         & 14,986 \\
                13460400002                 & 13,434.69         & 2,420 \\
                11554100002                 & 2,400.31          & 182 \\
                11722800003                 & 9,064.28          & 1,644 \\
                11717600002                 & 22,297.1          & 1,850 \\
                11722600002                 & 7,024.53          & 1,265 \\
                10652800002                 & 3,618.88          & 425.34 \\
				\textbf{Total}		    &\textbf{168,768.92}	&\textbf{30,659.54} \\
			\end{tabular}
			\caption{Transactions in \tabimporttrans which are in the same NCM and country of origin as selected target-transaction.}
		    \label{tab:example-all-transactions}
	\end{subtable}
	
	\begin{subtable}[b]{\linewidth}
	    \centering
			\begin{tabular}{lrr}
    			\textbf{State}  & \textbf{Value (USD)} 		& \textbf{Weight (kg)} \\ \hline
    			CEARA           &19,856     &4,055		\\
				ESPIRITO SANTO 	&112,762    &22,483		\\
				MINAS GERAIS 	&27,224     &2,937		\\
				SANTA CATARINA  &8,815      &1,161       \\
				SAO PAULO 		&112        &23		\\
				\textbf{Total}  &\textbf{168,769}		&\textbf{30,659} \\
			\end{tabular}
		\caption{States in the \tabsumbyncm dataset which are compatible with the  selected target-transaction.}
		\label{tab:example-possible-states}
	\end{subtable}
	
	\caption{Intermediate results used in the motivating example.}
	\label{fig:example-tables}
	
\end{table}

	If there is just one possible partition, or if in all possible partitions the target package would only fit in one destination, then 
	the adversary can be sure that she has correctly identified 
	the target transaction's destination state, and can proceed to Phase~3.
	Otherwise, the adversary will have more than one state and many companies left as possible importers, and she just makes a blind guess among those to infer the correct one.
	
	In the latter case, the adversary will not be absolutely certain of a 
	correct re-identification, but can still have some probabilistic confidence
	in the result.

	\textbf{The algorithm.}
	The algorithm used to find these valid partitions
	that link a transaction to its city/state employs integer programming techniques, 
	and is fully explained in Sec.~\ref{sec:implementation}. 
	Here we provide an overview of it in the context of our motivating example.
	
	For this algorithm, the records from 
	\tabimporttrans (\Table{tab:example-all-transactions})
	are represented as ``packages'', as shown in Fig.~\ref{fig:packages}. 
	Each package is represented by a rectangle whose
	width and height are the transaction's value and weight, respectively. 
	The totals by state from \tabsumbyncm
	(\Table{tab:example-possible-states}) are represented 
	as target coordinates in the plane, as shown in Fig.~\ref{fig:bins}. 
	The $x$ coordinate of these targets is positioned at the total value per state, and the $y$ coordinate is at the total weight per state.
	
	\begin{figure*}[!htbp]
	\centering
	\begin{subfigure}[b]{0.24\linewidth}
		\centering
		\begin{scriptsize}
			\begin{tikzpicture}[scale=0.7]
				\filldraw[fill=Green] (0,0) rectangle (0.338841, 0.42069)
				node[below, at start] {10653400001};
				\draw [|-|] (0, 0.7) -- (0.338841, 0.7)
				node[sloped,above,midway] {3,388.41 USD};
				\draw [|-|] (0.6, 0.42069) -- (0.6, 0)
				node[right, midway] {420.69 kg};
			\end{tikzpicture}
			
			\vspace{2mm}
			
			\begin{tikzpicture}[scale=0.7]
				\filldraw[fill=gray!20] rectangle (2.02168, 2.0913)
				node[midway, align=center] {123434\\00002};
				\draw [|-|] (0, 2.4) -- (2.02168, 2.4)
				node[sloped,above,midway] {20,216 USD};
				\draw [|-|] (2.3, 2.0913) -- (2.3, 0)
				node[sloped,above,midway] {2,091.3 kg};
			\end{tikzpicture}
			
			\vspace{2mm}
			
			\begin{tikzpicture}[scale=0.7]
				\filldraw[fill=gray!20] (0,0) rectangle (0.361888, 0.42534)
				node[below, at start] {10652800002};
				\draw [|-|] (0, 0.7) -- (0.361888, 0.7)
				node[sloped,above,midway] {3,618.88 USD};
				\draw [|-|] (0.7, 0.42534) -- (0.7, 0)
				node[right, midway] {425.34 kg};
			\end{tikzpicture}
		
			\vspace{2mm}
			
			\begin{tikzpicture}[scale=0.7]
				\filldraw[fill=gray!20] (0,0) rectangle (1.985621, 4.055)
				node[midway, align=center] {055665\\00003};
				\draw [|-|] (0, 4.3) -- (1.985621, 4.3)
				node[sloped,above,midway] {19,856.21 USD};
				\draw [|-|] (2.3, 4.055) -- (2.3, 0)
				node[sloped,above,midway] {4,055 kg};
			\end{tikzpicture}
		
		\end{scriptsize}
		\caption{Representation of de-identified transactions. The green one is the adversary's target.}
		\label{fig:packages}
	\end{subfigure}
	\hfill
	\begin{subfigure}[b]{0.24\linewidth}
		\centering
		\begin{scriptsize}
			\begin{tikzpicture}[scale=.7]
				\draw[->] (-0.1, 0) -- (3, 0) coordinate (x axis);
				\draw[->] (0, -0.1) -- (0, 3.2) coordinate (y axis);
				\filldraw (2.7224, 2.937) circle (2pt) node(a) {};
				\draw[dashed] (a) -- (a |- x axis);
				\draw[dashed] (a) -- (y axis |- a);
				\node[below] at (a |- x axis) {27,224 USD};
				\node[left] at (y axis |- a) {\rotatebox{90}{2,937 kg}};
				\node[above,align=left] at (a) {MINAS\\GERAIS};
			\end{tikzpicture}
			
			\vspace{2mm}
			
			\begin{tikzpicture}[scale=.7]
				\draw[->] (-0.1, 0) -- (2.1, 0) coordinate (x axis);
				\draw[->] (0, -0.1) -- (0, 4.3) coordinate (y axis);
				\filldraw (1.9856, 4.055) circle (2pt) node(a) {};
				\draw[dashed] (a) -- (a |- x axis);
				\draw[dashed] (a) -- (y axis |- a);
				\node[below] at (a |- x axis) {19,856 USD};
				\node[left] at (y axis |- a) {\rotatebox{90}{4,055 kg}};
				\node[above,align=left] at (a) {CEARA};
			\end{tikzpicture}
			
			\vspace{2mm}
			
			\begin{tikzpicture}[scale=.7]
				\draw[->] (-0.1, 0) -- (1.1, 0) coordinate (x axis);
				\draw[->] (0, -0.1) -- (0, 1.3) coordinate (y axis);
				\filldraw (0.8815, 1.161) circle (2pt) node(a) {};
				\draw[dashed] (a) -- (a |- x axis);
				\draw[dashed] (a) -- (y axis |- a);
				\node[below] at (a |- x axis) {8,815 USD};
				\node[left] at (y axis |- a) {\rotatebox{90}{1,161 kg}};
				\node[above,align=left,right] at (a) {SANTA\\CATARINA};
			\end{tikzpicture}
		\end{scriptsize}
		\caption{Representation of summaries by states.}
		\label{fig:bins}
	\end{subfigure}
	\hfill
	\begin{subfigure}[b]{0.24\linewidth}
		\centering
		\begin{scriptsize}
			\begin{tikzpicture}[scale=.7]
				\filldraw[fill=Green] (0,0) rectangle (0.338841, 0.42069);
				
				\filldraw[fill=gray!20] (0.338841, 0.42069) rectangle (2.360521, 2.51199)
				node[midway,align=center] {123434\\00002};
				
				\filldraw[fill=gray!20] (2.360521, 2.51199) rectangle (2.722409, 2.93733);
				\draw[->] (-0.1, 0) -- (3, 0) coordinate (x axis);
				\draw[->] (0, -0.1) -- (0, 3.2) coordinate (y axis);
				\filldraw (2.7224, 2.937) circle (2pt) node(a) {};
				\draw[dashed] (a) -- (a |- x axis);
				\draw[dashed] (a) -- (y axis |- a);
				\node[below] at (a |- x axis) {27,224 USD};
				\node[left] at (y axis |- a) {\rotatebox{90}{2,937 kg}};
				\node[above,align=left] at (a) {MINAS\\GERAIS};
				
			\end{tikzpicture}
			
			\vspace{2mm}
			
			\begin{tikzpicture}[scale=.7]
				\filldraw[fill=gray!20] (0,0) rectangle (1.985621, 4.055)
				node[midway, align=center] {055665\\00003};
				\draw[->] (-0.1, 0) -- (2.1, 0) coordinate (x axis);
				\draw[->] (0, -0.1) -- (0, 4.3) coordinate (y axis);
				\filldraw (1.9856, 4.055) circle (2pt) node(a) {};
				\draw[dashed] (a) -- (a |- x axis);
				\draw[dashed] (a) -- (y axis |- a);
				\node[below] at (a |- x axis) {19,856 USD};
				\node[left] at (y axis |- a) {\rotatebox{90}{4,055 kg}};
				\node[above,align=left] at (a) {CEARA};
			\end{tikzpicture}
		
			\vspace{2mm}
		
			\begin{tikzpicture}[scale=.7]
				\filldraw[fill=gray!20] (0,0) rectangle (0.769819, 0.91825);
				\filldraw[fill=gray!20] (0.769819, 0.91825) rectangle (0.881471, 1.16111);
				\draw[->] (-0.1, 0) -- (1.1, 0) coordinate (x axis);
				\draw[->] (0, -0.1) -- (0, 1.3) coordinate (y axis);
				\filldraw (0.8815, 1.161) circle (2pt) node(a) {};
				\draw[dashed] (a) -- (a |- x axis);
				\draw[dashed] (a) -- (y axis |- a);
				\node[below] at (a |- x axis) {8,815 USD};
				\node[left] at (y axis |- a) {\rotatebox{90}{1,161 kg}};
				\node[above,align=left,right] at (a) {SANTA\\CATARINA};
			\end{tikzpicture}

		\end{scriptsize}
		\caption{Representation of a valid solution.}
		\label{fig:example-valid-solution}
	\end{subfigure}
	\hfill
	\begin{subfigure}[b]{0.24\linewidth}
		\centering
		\begin{scriptsize}
			\begin{tikzpicture}[scale=.7]
				\filldraw[fill=Green] (0,0) rectangle (0.338841, 0.42069);
				\filldraw[fill=gray!20] (0.338841, 0.42069) rectangle (2.324462, 4.447569)
				node[midway,align=center] {055665\\00003};
				
				\draw[->] (-0.1, 0) -- (3, 0) coordinate (x axis);
				\draw[->] (0, -0.1) -- (0, 3.2) coordinate (y axis);
				\filldraw (2.7224, 2.937) circle (2pt) node(a) {};
				\draw[dashed] (a) -- (a |- x axis);
				\draw[dashed] (a) -- (y axis |- a);
				\node[below] at (a |- x axis) {27,224 USD};
				\node[left] at (y axis |- a) {\rotatebox{90}{2,937 kg}};
				\node[above,right,align=left] at (a) {MINAS\\GERAIS};
			\end{tikzpicture}
			
			\vspace{2mm}
			
			\begin{tikzpicture}[scale=.7]
				\filldraw[fill=Green] (0,0) rectangle (0.338841, 0.42069);
				\filldraw[fill=gray!20] (0.338841, 0.42069) rectangle (2.324462, 4.47569)
				node[midway, align=center] {055665\\00003};
				
				\draw[->] (-0.1, 0) -- (2.1, 0) coordinate (x axis);
				\draw[->] (0, -0.1) -- (0, 4.3) coordinate (y axis);
				\filldraw (1.9856, 4.055) circle (2pt) node(a) {};
				\draw[dashed] (a) -- (a |- x axis);
				\draw[dashed] (a) -- (y axis |- a);
				\node[below] at (a |- x axis) {19,856 USD};
				\node[left] at (y axis |- a) {\rotatebox{90}{4,055 kg}};
				\node[above,align=left] at (2.324462, 4.47569) {CEARA};
			\end{tikzpicture}
		
			\vspace{2mm}
			
			\begin{tikzpicture}[scale=.7]
				\filldraw[fill=gray!20] (0,0) rectangle (0.769819, 0.91825);
				\filldraw[fill=Green] (0.769819, 0.91825) rectangle (1.10866, 1.33894);
				\draw[->] (-0.1, 0) -- (1.1, 0) coordinate (x axis);
				\draw[->] (0, -0.1) -- (0, 1.3) coordinate (y axis);
				\filldraw (0.8815, 1.161) circle (2pt) node(a) {};
				\draw[dashed] (a) -- (a |- x axis);
				\draw[dashed] (a) -- (y axis |- a);
				\node[below] at (a |- x axis) {8,815 USD};
				\node[left] at (y axis |- a) {\rotatebox{90}{1,161 kg}};
				\node[right,align=left] at (1.10866, 1.33894) {SANTA\\CATARINA};
			\end{tikzpicture}
			
		\end{scriptsize}
		\caption{Representation of an invalid solution.}
		\centering
		\label{fig:example-invalid-solution}
	\end{subfigure}
	\caption{Graphical representation of the package-allocation algorithm for Phase 2 of the attack, linking transactions to states.}
	\label{fig:algorithm-packages}
	\end{figure*}
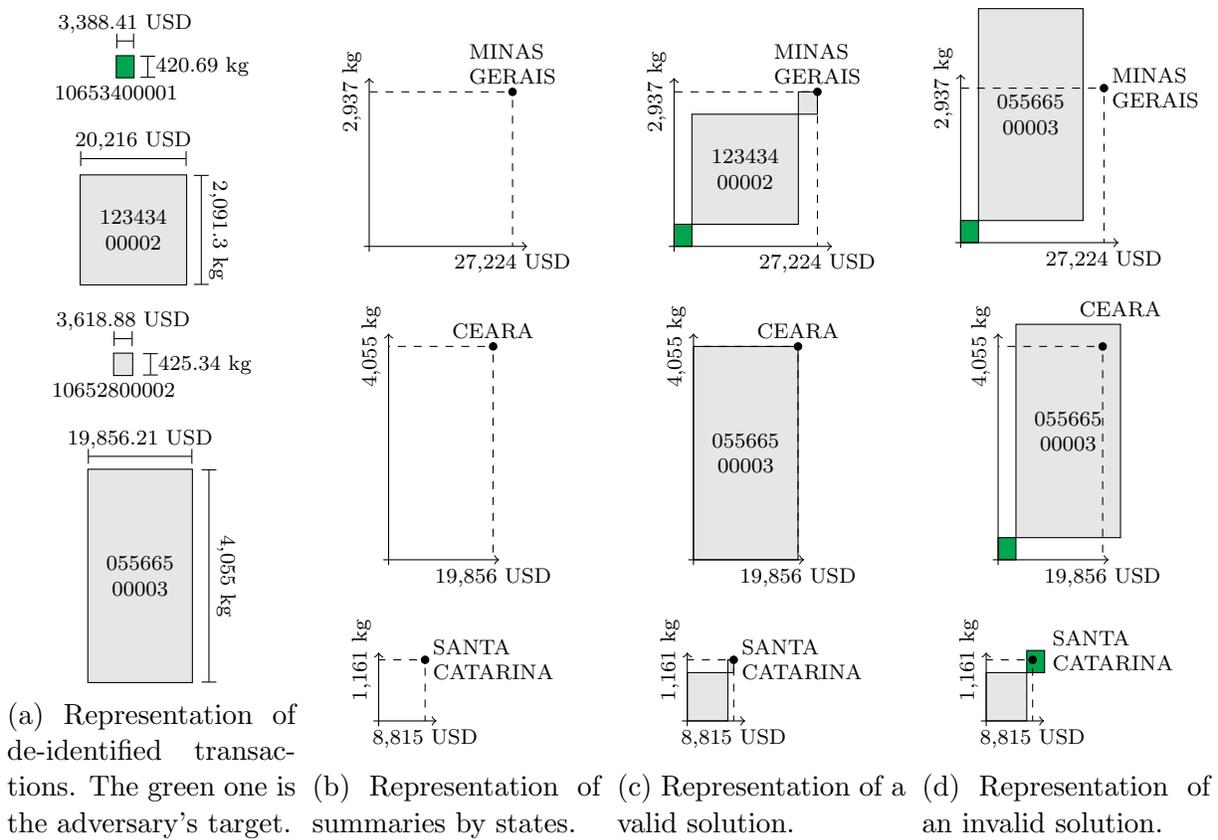

	From this setup, the adversary tries to fit all the ``packages'' to the planes without violating the boundaries represented by the dashed lines. 
	The packages must be stacked by their corners as shown in Figs.~\ref{fig:example-valid-solution} and \ref{fig:example-invalid-solution}. 
	This arrangement guarantees that in a valid solution the sum of packages' values will be equal to the total observed in the state represented by the corner point. 
	Figs.~\ref{fig:example-valid-solution} and \ref{fig:example-invalid-solution} depict valid and invalid solutions, respectively. 
	It is easy to see that in the valid one, the sums of the packages with \texttt{NO\_ORDEM} 12343400002, 10653400001 and 10652800002 match the total values of 
	the state of MINAS GERAIS, whereas the package with 
	\texttt{NO\_ORDEM} 05566500003 alone matches the totals of the state of CEARA.
	If, however, we swap the packages with \texttt{NO\_ORDEM} 12343400002 and
	05566500003, we no longer have a valid solution for either state.

	Moreover, it is possible to test if there is any other valid allocation in which the package with \texttt{NO\_ORDEM} 10653400001 is placed in a state other than MINAS GERAIS. 
	In fact, for this example there is not. 
	This is conclusive evidence that this package must have been purchased by a company based in MINAS GERAIS.
	
	With that, the adversary has successfully linked the target transaction to a single state, and she can move on to the next phase of the attack.
	
	\textbf{Attack Phase 3: Trying to link the target transaction to a city once the state is known.}
	Now that the adversary knows that the target transaction was made by a company based in MINAS GERAIS she can return to a similar problem to the one she faced in Phase 1. Before knowing the state there were too many possible solutions and the problem was intractable. But by incorporating restrictions found in Phase 3, she is able to reduce the number of possible allocations and (potentially) make the problem tractable. The algorithm used to solve this new, reduced instance is the same as the one explained for the case of the states. If this reduced instance of the problem is tractable, the adversary may link the target transaction to a specific city.
	In our example, the adversary is able to identify the city of OURO BRANCO, in the state of MINAS GERAIS, as the location of the target transaction.
	
	\textbf{After match: Inferring importer from city.}
	After the city in which the transactions have occurred has been identified, the adversary proceeds to infer which company based on that city may be the importer.
	In our example, a simple inspection of the \dataset{Importers} dataset
	reveals to the adversary that there is only one company in the city of OURO BRANCO, which leads to an immediate re-identification of the importer.
	(For ethical reasons, we omit the affected company's name.)
	
    Hence, in our example the adversary has a high confidence that she achieved a deterministic re-identification.
    This leads to the inference of the following sensitive
	information:
		(i) the importer purchased 420.69 kg of a specific type of cotton cloth; and 
		(ii)  that purchase cost  3,388.41 USD, so the average price is approximately 8.06 USD/ton.
	In our view, this is a clear violation of the importer's 
	fiscal secrecy within Brazilian legal framework.\footnote{Article 198 of the National Tax Law 5.172, October \nth{25} of 1966 and Article \nth{8} of Ordinance 7,017 of March \nth{11} of 2020.}
		
	In the following sections we formalize the general version of this
	attack, and explain how to compute its overall leakage of secret fiscal information.


\section{Attack model}\label{sec:model}

We now formalize our model, generalizing from the previous section's example, using principles from \emph{quantitative information flow} (QIF)~\cite{ScienceOfQIF},
which has sound information- and decision-theoretic grounds to model an adversary's knowledge, goals, and capabilities, as well as the
leakage of information caused by an attack. This formalism is helpful for 1) measuring the attack effectiveness and reach and 2) to provide a baseline that future works can use to compare different defensive strategies against this attack.

\subsection{Channels, secrets and adversaries}

Using QIF we model adversaries who employ \emph{Bayesian inference} on 
data release mechanisms modelled as noisy information-theoretic channels. The adversary updates her prior knowledge about sensitive values or ``secrets'' using Bayes' rule, thereby inducing a measure of ``channel leakage'', which incorporates the difference between her prior and posterior knowledge as well as her goals. We now make these ideas more precise.

A \emph{secret} models the information sought by the adversary; the set of possible secrets is denoted by $\calx$.
A \emph{prior} models the adversary's prior knowledge about the secrets (ie., before the attack is even performed), which can come from various sources, including other datasets, media,
or simply common-sense about foreign trade.
Prior knowledge is formalized as a probability distribution on possible 
secret values denoted $\pi \in \mathbb{D}{\calx}$.
(Here, $\mathbb{D}{\calx}$ is the set of probability distributions over the elements of set $\calx$.)

A \emph{channel} $C$ takes inputs (secrets) $x \in \calx$ and produces an observation $y \in \caly$ according to some distribution in $\mathbb{D}{\caly}$. When $\calx$ and $\caly$ are discrete we write $C$ as a matrix whose element $C_{x,y}$ is the probability of producing $y \in \caly$ when the input is $x \in \calx$. Rows in $C$ are distributions over $\caly$ and hence are 1-summing; if $C$ is deterministic then every entry in $C$ is $0$ or $1$, with each row containing exactly one value $1$.

Under the \emph{g-leakage} framework~\cite{ScienceOfQIF}, the adversary's goal is modelled with a \emph{gain function}, and her strategy is to maximize her expected gain using her knowledge of the channel. In this work we adopt the ``Bayes vulnerability'' gain function which simply models an adversary whose goal is to guess the secret in a single attempt. We describe the attack in terms of the secret's ``vulnerability'' with respect to the adversary in question.
Given the adversary's prior $\pi \in \mathbb{D}{\calx}$, the prior vulnerability of the secret is given by
\[
	\vbayes(\pi) ~=~ \max\limits_{x \in \calx} \pi(x) ~,
\]
and it measures how likely it is that the adversary can correctly guess the secret value in one try.
Given access to the channel $C$, the posterior vulnerability (i. e., after the attack is performed) can be computed as\footnote{ See ~\cite[Ch 5 Thm 5.15]{ScienceOfQIF} for details.}
\[
	\vbayes(\pi, C) ~=~ \sum\limits_{y \in \caly} \max\limits_{x \in \calx} \pi(x){\times}C_{x,y} ~,
\]
and it represents the expected probability of the adversary's 
correctly inferring the secret value in a single attempt after having observed 
the output of the channel.

Finally, we can compute the leakage of the channel (multiplicatively) by computing how much the channel increases the adversary's success, as
\begin{equation}\label{eq:bayesleak}
	\bayesleak(\pi,C) ~=~ \frac{\vbayes(\pi,C)}{\vbayes(\pi)}~.
\end{equation}

\subsection{Modelling the reconstruction}

In our scenario, the adversary's goal is to identify the company behind a transaction chosen as the target, using the information from the data release. We model the \emph{secrets} $\calx$ as the set of possible importers and the adversary's \emph{prior} as uniform on $\calx$. That is, we assume a fairly weak adversary who lacks any common knowledge on trade transactions.
We model the data release mechanism \emph{for the target transaction}~\footnote{We would model the entire data release as a channel if the adversary's goal was to learn information about \emph{any} transaction.} as a channel~\footnote{In fact what this encodes is the \emph{correlation}, which is a joint distribution, but this can be factorized into a channel by taking appropriate marginals.} which maps importers (secrets) to cities (observations). 

Before the attack, the channel is ``noisy'' from the adversary's perspective, as she lacks precise information about the correlation between the target transaction and the city of import. The  attack's goal is to reconstruct (more precisely) this correlation, and therefore to improve her likelihood of guessing the correct importer, by exploiting the information released in the \tabsumbycity and \tabsumbyncm datasets.

\subsection{Measuring the leakage of information caused by the attack}
We now return to the example of \Sec{sec:overview-example} to describe how the channel is constructed and to examine the overall leakage caused by the attack, which is quantified by comparing the 
	adversary's measure of success before and after the attack was 
	conducted, as per Eqn~(\ref{eq:bayesleak}).
	
	\Table{tab:chance_of_success_by_phase} presents the adversary's probability
	of correctly re-identifying the importer of the target transaction (10653400001) at each phase of the attack.
	Before the attack, this probability is uniform over all importers, and thus the prior vulnerability is simply $\nicefrac{1}{18,430}$, as our
	ill-informed adversary could only perform a blind guess over the all equally likely $18,430$ importers.
	At this point the \emph{channel} representing the data release has as its inputs all $18,430$ importers and as its outputs all possible cities.
	
	At the end of Phase 1, after having access to the \tabsumbycity,  the adversary knows that the target transaction could have gone to one of the 19 cities that imported from China in the SH4 5208.
	This reduces the number of possible importers from the initial 18,430 to 4,923. In our channel model, this has the effect of updating the adversary's \emph{prior} to become uniform over the remaining $4,923$ importers (eliminating the remainder options). The (posterior) vulnerability after this update is $\sum_{x} \pi(x) C_{x,y} = \nicefrac{1}{4,923}$. The information leakage caused by this phase of the attack is 
	$\bayesleak(\pi,C) = \nicefrac{(\nicefrac{1}{4,923})}{(\nicefrac{1}{18,430})} \approx 3.7$,
	which means that the adversary's probability of correctly re-identifying the target transaction's importer has increased by a factor of $3.7$ when compared to before the attack.
	
	At the end of Phase 2, the adversary knows that the transaction went to one of the 3 cities of the state of MINAS GERAIS that imported in the SH4 from China. The \tabimporters dataset shows that BELO HORIZONTE has 276 registered importers, EXTREMA has 49, and OURO BRANCO has 1. This reduces the number of possible importers from 4,923 to 326. As before, the posterior vulnerability increases to $\nicefrac{1}{326}$, and the information leakage caused by this phase is $\nicefrac{(\nicefrac{1}{326})}{(\nicefrac{1}{4,923})} \approx 15$. 
	Finally, at the end of Phase 3, the adversary knows that the transaction went to the city of OURO BRANCO, which has only one registered importer. Thus the posterior vulnerability is now $1$ (i. e. the adversary's probability of guessing the importer is $1$) and the leakage caused by the final phase is $326$.

	\begin{table}[!htbp]
	    \centering
		\begin{tabular}{
				 	l
				 	c
				}
				\textbf{Phase} 	& \textbf{Probability of correct re-identification}
				\\
				\hline
				Prior & $\nicefrac{1}{18,430} = 0.005\%$ 
				\\
                Phase 1  & $\nicefrac{1}{4,923} = 0.02\%$ 
                \\
                Phase 2  & $\nicefrac{1}{326} = 0.3\%$ 
                \\
                Phase 3  & $\nicefrac{1}{1} = 100\%$ 
                \\
		\end{tabular}
		\caption{Measure of the adversary's success by phase. Notice that by the end of Phase 1 the adversary's probability of correct re-identification
		increases by a factor $3.7$ over the prior.
		At the end of Phase 3, the increase factor is of $326$
		over the previous phase, and of $18,430$ over the prior.}
		\label{tab:chance_of_success_by_phase}
	\end{table}

\section{The attack algorithm}\label{sec:implementation}

In this section we introduce the mathematical modelling of the task of placing transactions (packages) into cities or states (bins) as an integer linear program (ILP). In this setting, we aim to introduce variables whose $1$ or $0$ value, after the execution of the algorithm, will indicate whether a package is or is not placed in a bin. The motivation to recourse to this framework is twofold.

First, a brute force attack that attempts to test all possible allocations and verify whether they satisfy validity constraints given by value and weight is intractable under current computational resources even for small examples. For instance, the example described in \Sec{sec:overview-example} yields more than six billion ways of distributing 14 transactions in 5 states.

In general, if one iterates over all possible allocations of $n$ packages in $m$ possible bins, this yields $m^n$ possible allocations. A slightly larger instance than the one considered in \Sec{sec:overview-example}, with $n = 40$ and $m = 12$, yields as many as roughly $10^{43}$ possibilities. Clearly, several of these possibilities do not need to be tested as they can be identified as sub-cases of already infeasible allocations. For instance, in Tbls.~\ref{tab:example-possible-states} and \ref{tab:example-all-transactions}, package 02898100002 can only fit in the state of ESPIRITO SANTO, thus all further tested attempts should be made with this allocation set, which eventually leads to the backtracking paradigm for testing exponentially many possibilities. 

Upon modelling the problem as an integer linear program, standard techniques to solve such programs, implemented in what is typically referred to as an \textit{ILP solver}, already employ the backtracking paradigm in conjunction with other geometric methods, yielding branch-and-bound or branch-and-cut strategies, for example. In principle, a suitable backtracking implementation should be marginally faster \textit{if an exact allocation exists}.
However, this may come at the cost of being much less flexible to modifications later observed in the data, and also of being harder to implement. 

The second reason why we argue the ILP formulation is more suited than brute force, or the more sophisticated backtracking methods, is that the data typically does not allow for an exact allocation. As errors can be introduced in the data either involuntarily or by design (suppressing decimals, for instance), one should not expect the values in \Table{tab:example-all-transactions} to sum up exactly to the ones seen in \Table{tab:example-possible-states}, for example. Therefore one will typically wish to find an allocation that minimizes error --- as it happens, the branch-and-bound method to solve ILPs is essentially backtracking applied when a target linear function needs to be minimized.

Finally, note that it could still be the case that a clever approach could find a package allocation more quickly, upon exploring a possible specific combinatorial structure of the problem. We believe that is highly unlikely, however, as we can show that the formal problem underlying our modelling is NP-complete (Thm.~\ref{thm:npcomplete} ahead).

\subsection{Mathematical modelling and complexity}\label{sec:implementation-modelling}

We now mathematically model the problem of allocating transactions (packages) into cities or states (bins). This shall be referred to as the \texttt{PACKAGE ALLOCATION} decision problem, whose instances are as follows. The mathematical formulation shown in \Fig{fig:ilp-formulation} will be helpful for any future replication or expansion of this work, and Theorem~1 is important to base our decision of using ILP and not some other custom tailored algorithm.

\paragraph*{Instance:}
	\begin{equation}
	    \mathfrak{I} = (P, k, w, B, c, \rho, \beta,  \varepsilon) \label{eq:instance}
	\end{equation}
	is the input tuple for the PACKAGE ALLOCATION problem, where:
	\begin{enumerate}
		\item $P = \{p_1, p_2, ..., p_n\}$ is a set of $n$ packages.
		\item $k$ is the number of attributes each package has. In principle, each attribute is a non-negative rational number. In the example shown in \Sec{sec:overview-example} there are two attributes: weight and value.
		\item $w : P \to \mathbb{Q}_+^k$ is a function that determines for each package its attribute set. The $i$th attribute of package $p$ will be denoted by $w(p;i)$.
		\item $B = \{b_1, b_2, ..., b_m\}$ is a set of $m$ bins.
		\item $c : B \to \mathbb{Q}_+^k$ is a function that determines for each bin its capacity for each attribute.
		\item $\rho \in P$ is a special, target package.
		\item $\beta \in B$ is a special, target bin for the target package.
		\item $\varepsilon : B \to \mathbb{Q}_{+}^k$ is a non negative tolerance for each bin and for each attribute.
	\end{enumerate}

An instance is a \texttt{YES}-instance if and only if there exists a function $X : P \times B \to \{0,1\}$ (naturally identified with a $|P| \times |B|$ matrix with $0/1$ entries) satisfying:
\begin{enumerate}
    \item[(a)] For each $p \in P$,
    $\sum_{b \in B} X_{pb} = 1$.
    (Each package is allocated to exactly one bin).
    \item[(b)] For each $i \in \{1,...,k\}$, and for each $b \in B$, 
    $\sum_{p \in P} X_{pb} w(p;i) \leq c(b;i) + \varepsilon(b;i),$
    and
    $\sum_{p \in P} X_{pb} w(p;i)  \geq c(b;i) - \varepsilon(b;i)$.
    (Bins' capacities are filled up to tolerance for all attributes).
    \item[(c)] $X_{\rho \beta} = 1$.
    (Package $\rho$ is allocated at bin $\beta$).
\end{enumerate}

It is immediate to verify that \texttt{PACKAGE ALLOCATION} is in NP, as any given function $X$ can be verified in time polynomial to the size of the instance. More interestingly, we have the following result.

\begin{restatable}[]{theorem}{thmnpcomplete}
\label{thm:npcomplete}
\texttt{PACKAGE ALLOCATION} is NP-complete.
\end{restatable}

\begin{proof}
An instance of the \texttt{SUBSET SUM} problem is a set of natural numbers $N = \{n_1,...,n_\ell\}$ and a natural number $S$. It is an \texttt{YES}-instance if and only if a subset of $N$ sums to exactly $S$. This problem is well known to be NP-complete (see for instance \cite[Theorem 8.23]{kleinberg2006algorithm}). Given an instance of \texttt{SUBSET SUM} as above, generate $\ell$ instances of \texttt{PACKAGE ALLOCATION}, each of which contains $\ell$ packages (one for each number), only one attribute (the value of the natural number), two bins, one of capacity $S$ and one of capacity $\sum n_i - S$, $\varepsilon \equiv 0$, fixing $\beta$ to be the bin of capacity $S$ and making $\rho$ vary over all packages. It is immediate to verify that an instance of \texttt{SUBSET SUM} is \texttt{YES} if and only if there is at least one package $\rho$ for which the corresponding instance of \texttt{PACKAGE ALLOCATION} is \texttt{YES}. As this is a polynomial reduction, the result follows.
\end{proof}

Since a tailored efficient algorithm is likely nonexistent, we present
in \Fig{fig:ilp-formulation}
an exact integer linear program 
formulation. We use Instance \eqref{eq:instance} and allow $M$ to be a large enough number exceeding the sums of attributes.

\begin{figure}[!htbp]
\noindent\textbf{\texttt{PACKAGE ALLOCATION} ILP formulation:}
	\begin{align*}
		\text{minimize } 	& y_1 + ... + y_k. \\
		\text{subject to } 	& \\ 
		(a)\quad & \sum_{b \in B} X_{pb} = 1, \quad \text{for all }p \in P; \\
		(b)\quad & \sum_{p \in P} X_{pb} w(p;i) \leq c(b;i) + \varepsilon(b;i) + y_i M, \\ & \qquad \qquad \text{for all }b \in B, i \in [k]; \\
		(c)\quad & \sum_{p \in P} X_{pb} w(p;i)  \geq c(b;i) - \varepsilon(b;i) - y_i M, \\ & \qquad \qquad \text{for all }b \in B, i \in [k]; \\
		(d)\quad & X_{\rho\beta} = 1 \\
		(e)\quad & X_{pb} \in \{0,1\}, \quad \text{for all }p \in P, b \in B ;\\
		(f)\quad & y_i \in \{0,1\}, \quad \text{for all }i \in [k].
	\end{align*}
	\caption{Formulation of the \texttt{PACKAGE ALLOCATION} problem in ILP.}
	\label{fig:ilp-formulation}
\end{figure}

It is immediate to verify that the optimum solution of the program above is equal to $0$ if and only if the given instance is a \texttt{YES} instance. Moreover, the optimum solution, if not equal to $0$, indicates which bins and attributes had to be violated in a best possible allocation, possibly helping to indicate where inconsistencies 
in the data occur, as we discuss in \Sec{sec:implementation-rounding}. 

Finally, for the sake of intelligibility, below we discuss a modified model in which only one attribute is considered, rendering the formulation above significantly simpler. In this case, the objective function simply minimizes a real variable $z$, understood to model the maximum bin violation, thus restrictions of type (b) and (c) are replaced,  respectively, for
\[
    \sum_{p \in P} X_{pb} w(p){\leq} c(b){+}z
    \quad \text{and} \quad
    \sum_{p \in P} X_{pb} w(p){\geq}c(b){-}z.
\]
In this case, the instance is accepted if and only if the optimum $z$ is below a given threshold, henceforth introduced as $\nu$. One advantage of this model with only one attribute is that, upon removing the constraint for the special package and target bin, the solver will always return a feasible solution. This can be interpreted as a best possible allocation,
and therefore allows for a good measure of the quality of the data. For example, data which has suffered artificial modifications will typically result in best possible allocations with large errors.

The ILP presented above can be readily implemented in any solver, as we discuss in \Sec{sec:experiments}.

\subsection{Exact algorithm}\label{sec:implementation-exact}

In our attack model, we assume the adversary wants to know all the possible cities (or states) where the target transaction ($\rho$) might be placed.

There are two ways to achieve that. 
One of them requests the solver to print all feasible solutions whose objective value has achieved a given threshold. Not all solvers will handle this task particularly efficiently. 

A second direct approach is as follows. First, request one valid allocation without restricting the bin where the package $\rho$ can be placed, within the overflow tolerance $\nu$.\footnote{The adversary knows that the summary statistics are produced from the microdata, so they can be sure that at least one valid allocation exists.} 
After the first iteration, we retrieve the bin $b$ where the package $\rho$ was allocated and add as a restriction that $X_{\rho b} = 0$ for a next run of the \texttt{PACKAGE ALLOCATION} solver routine, thus finding possible new allocations for which package $\rho$ is placed in a different bin. The procedure then loops, adding restrictions until it becomes infeasible. The set $R \subseteq B$ containing the possible bins for the package $\rho$ is returned at the end.
For later reference, we will call this our \texttt{ATTACK ALGORITHM}.

Notice how $R$ contains all mathematically possible bins on which the package $\rho$ would fit. If $\rho$ is to be placed in a bin not present in $R$ that would necessarily create a partition of packages in which at least one of the bins capacity would be greatly overflown (beyond the tolerance). The algorithm's exactness is translated to an assured re-identification given the assumption that \tabimporttrans contains the same transactions used to produce the total shown at \tabsumbycity and
\tabsumbyncm (which is ensured by the Brazilian government with high degree of certainty \cite{NotaTecnicaSecex}).


\subsection{Dealing with imperfect data}\label{sec:implementation-rounding}

Although the foreign-trade statistics disclosed in the four datasets presented in Sec.~\ref{sec:overview-datasets} derive from the same source, some methodological differences produce divergences in the data.

\begin{enumerate}
	\item \textbf{Numeric precision}: The \tabimporttrans reports values with up to 2 decimal places for currency and 5 decimal places for weight. \tabsumbycity and \tabsumbyncm round all values to integers.
	\item \textbf{Data suppression}: The \tabimporttrans does not disclose those NCMs in which less than 4 importers operated, in an attempt to preserve privacy. It is, however, easy to find which NCMs were omitted by observing the differences (in NCMs) between \tabimporttrans and \tabsumbyncm.
	\item \textbf{Outliers treatment}: \tabsumbycity and \tabsumbyncm datasets are intended for statistical purposes and international comparability. Transactions that are caught in internal controls might be discarded from the official publications, and this process is opaque to external scrutiny.
	\item \textbf{Transactions inclusion criteria}: On April \nth{7}, 2021 the Ministry of Economy changed its foreign commerce statistics methodology and started including some imports classified as ADMINISTRATIVE or SPECIAL in their reports. The information used for the selection criteria is not present in the datasets~\cite{NotaTecnicaSecex} and we are unaware of publicly available datasets that have this information.
\end{enumerate}

Divergences caused by rounding values are easy to deal with, as our tolerance parameters can be adjusted accordingly. Divergences caused by the data suppression, treatment of outliers, and the different inclusion criteria, however, do have a detrimental effect on our attack algorithm, even though only the data suppression technique has confidentiality as a goal. 

Despite the fact that this kind of data manipulation may act as a kind of protection 
against database reconstruction attacks, 
it was not introduced with the explicit intent of increasing privacy, nor was it applied in a controlled manner.
Importantly, it lacks the formal guarantees of techniques based
on, e. g., differential privacy~\cite{DifferentialPrivacy}, and consequently, as we shall see,
did not meaningfully protect against our attack.

For instance, by grouping the \tabimporttrans, \tabsumbycity, and \tabsumbyncm datasets by SH4 and NCM codes, and comparing the total USD values we found the following for January/2021:

\begin{enumerate}
	\item Using the \tabsumbyncm dataset as the official source, the total imports amount to 15.2 billion USD spread across 7,062 unique NCMs. 
	\item There is no divergence between the \tabsumbyncm and \tabsumbycity datasets when they are grouped by SH4, country of origin, and state.
\end{enumerate}

As for differences between the \tabimporttrans and \tabsumbyncm datasets in the month of January/2021, we can highlight:

\begin{enumerate}
	\item \tabimporttrans suppressed the information of all transactions among 2,277 NCMs. This amounts to 1.5 billion USD (10\% of the total imported in the period of interest).
	\item For 1,672 NCMs which amount to 5.3 billion USD (35\% of the total imported)
	the differences for values are between 2 and -2 USD. This might be caused just by rounding or small transactions included or excluded by each publisher.
	\item For 284 NCMs which amount to 1.8 billion USD (12\% of the total imported)
	the totals obtained from \dataset{De-identified transactions} values is greater than or equal to 2 USD, when compared to \dataset{Summary by State}. Those differences might be related to the outliers treatment.
	\item For the remaining 1,211 NCMs which amount to 6.4 billion USD (42\% of the total imported)
	the totals obtained from \dataset{De-identified transactions} values is lower than or equal to 2 USD. Those differences are mainly related to the inclusion of some ADMINISTRATIVE transactions in the releases.
\end{enumerate}

The different methodologies used pose some challenges to the applicability of the \texttt{ATTACK ALGORITHM}. The main one is that the algorithm operates under the premise that each package must go to exactly one bin. However, as discussed above, there are some cases in which the Ministry of Economy deliberately discards some of the transactions from their summaries. In such cases, if we included those transactions in the algorithm we would end up with an allocation that does not reflect reality. To avoid this kind of problem, we need to 
identify the cases in which the Ministry of Economy has discarded some transactions and 
either: (i) use the \texttt{ATTACK ALGORITHM} to identify which transactions were discarded before trying to use the same algorithm to track the transaction importer; or (ii) simply not solve these cases because of the extra complexity. In this work, we only considered the second alternative.

Although this second option is simpler, we still need, in both cases, to correctly identify those cases where the Ministry of Economy might have discarded some transactions from the summaries. This task would be trivial if we could model our problem without considering any rounding errors. The challenge therefore lies in how to distinguish divergences caused by discarded transactions from those caused by the rounding methodology.

If the rounding mechanism is unknown,\footnote{We tried 1) rounding the transactions' values to the closest integer, and then summing, and 2) adding up all transactions and rounding the totals. In neither case were we able to replicate the rounding methodology used by the Ministry of Economy.} it is possible to model it as a random error with a uniform distribution. If we are interested only in the differences produced by rounding mechanisms, a function that rounds a number to an integer will produce similar results to adding a small random number to it.

For instance, assume $n$ transactions with rational values $v_i \in \mathbb{Q}$ were reported instead as their nearest integer. We can model those values as $v_i = k_i + \varepsilon_i$, where $k_i$ is an integer and $-0.5 < \varepsilon_i \leq 0.5$. Assuming the decimals in the rational numbers follow a uniform distribution, then so do the variables $\varepsilon_i$. We can therefore write $\sum^{n}{v_i} - \sum^{n}{k_i} = \sum^{n}{\varepsilon_i}$, and the maximum value that $\sum^{n}{\varepsilon_i}$ can assume is $n \times 0.5$. This could be used as a threshold 
to distinguish differences due to rounding from those due to transaction exclusion. 

Note, however, that this limit 
is already too loose. As $n$ grows, it becomes increasingly unlikely for $\sum^{n}{\varepsilon_i}$ to be equal to $n \times 0.5$. In fact by the Central Limit Theorem~\cite{EstatisticaBasica}, the probability that $\sum^{n}{\varepsilon_i} > n \times \nicefrac{2.33}{\sqrt{12n}}$ is lower than 1\%. That is the threshold 
we adopted for telling apart which differences in the figures reported were due to rounding and which were due to some exclusion criteria adopted by the Ministry of Economy.

For example, if we have 3 transactions whose original values sum up to 5.7, but the reported sum is 4.5, the actual difference is 1.2, and the threshold 
is 1.165. Hence, if our premises are right, there is less than a 1\% chance that this difference is due to just rounding, and the more probable hypothesis is that one of the transactions was not included in the total report. In such cases we did not try to re-identify the transactions.


\section{Experimental evaluation}\label{sec:experiments}

In this section we evaluate how our \texttt{ATTACK ALGORITHM} from
\Sec{sec:implementation-exact} performs in a large-scale 
re-identification attack on the \tabimporttrans 
dataset from January/2021.
This case study is conducted with the treatment of imperfect data
discussed in \Sec{sec:implementation-rounding}, and the
ILPs from \Sec{sec:implementation-modelling}  are used as sub-routines.

\subsection{Experimental set-up and computational cost}\label{sec:experiments-cost}
Our \texttt{ATTACK ALGORITHM} was implemented in Python 3.7
and uses the optimization package Pyomo 6.0.1 (for the solver routine), as well as the Pandas 1.3.1 package. 
The solver used was IBM$^\text{\textregistered}$ ILOG$^\text{\textregistered}$ CPLEX$^\text{\textregistered}$ version 12.9.0.0. 
The whole experimental evaluation was run in a personal computer Intel$^{\text{\textregistered}}$ Core$^{\text{TM}}$ i7-8550U CPU @ 1.80GHz with 32GB of ram memory.

We set 1 minute as the threshold for considering an instance
solved in reasonable time. 
Through all attack phases, 84,698 instances of the problem
were built, and of, those, 69,244 (82\%) were successfully 
solved within the time limit.
The remaining 15,454 (18\%) either failed to return a result within the time constraint or were ignored due to the errors produced by the rounding mechanism being over the threshold explained in \Sec{sec:implementation-rounding}. This took in total almost 62 hours of active processing.

\Table{tab:solving_stat} shows that the average solving time (for the instances that were solved) was 1 second. In fact, more than 99\% of the solved instances (around ${68,000}$) were solved in under 28 seconds.
Finally, \Fig{fig:complexity_vs_solving_time} shows how the average solving time (in seconds) increased with the instance complexity (defined here as $\log_{10}{m^{n}}$ where $m$ is the number of bins and $n$ the number of packages in the instance. This value gives a sense of the ILP search space, higher numbers represent bigger search spaces). 

\begin{table}[!htbp]
	\centering
	\begin{tabular}{lccc}
		\textbf{Status}&\textbf{Count}&\textbf{Avg. Complexity}&\textbf{Avg. Solving time (s)}
		\\ \hline
		Unsolved& 15,454 		&  136.4 & - 		\\
		Solved	& 69,244 		&  3.73 & 1.39 	\\
		Total & 84,698 & 28 & - \\
	\end{tabular}
	\caption{Statistics about solved and non-solved instances. The Avg. complexity is the simple mean of $log_{10}{m^n}$ across all instances.}
	\label{tab:solving_stat}
\end{table}

\begin{figure}[!htbp]
	\centering
	\includegraphics[width=0.7\textwidth]{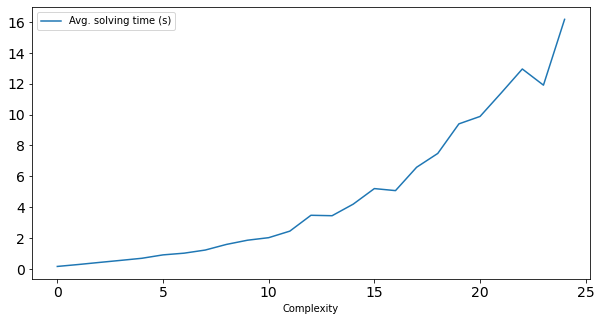}
	\caption{Average solving time by complexity ($\log_{10}{m}^{n}$), where $m$ is the number of bins and $n$ the number of packages.}
	\label{fig:complexity_vs_solving_time}
\end{figure}

\subsection{Results and analyses}\label{sec:experiments-results}
After the complete execution of our attack
(i. e., all 3 Phases described in \Sec{sec:overview-example}),
we were able to re-identify with certainty the 
importer's city for 138,413 transactions
out of a total of 817,468 in the dataset
(amounting to 6.1 billion dollars, or 40\% of the total value imported in the period of interest).
For 2,003 of such transactions (amounting to 137.3 million dollars, or 0.9\% of the total value imported), there is only one importer in the recovered city. Thus, we were able to learn with high confidence sensitive details of 348 companies (2\% of the total that imported something in the period of interest). Tbl.~\ref{tab:most-violated-cities} shows the top 5 cities (from a total of 348) in which the total value of transactions re-identified with absolute certainty was largest.

Tbl.~\ref{tab:final-results} summarizes
the adversary's success in each attack phase and step.
We provide the probability of re-identification
of a randomly selected transaction, as well as
the corresponding multiplicative leakage (i. e.,
the factor by which this probability increases wrt.\ the prior one).
In particular, after the conclusion of Phase 3, Step 2, the adversary's probability of re-identification is 0.9\%, corresponding to an increase of $168\times$ over the prior risk. Although this probability appears to be small, observe that it corresponds to a significant number of transactions with a combined value of 137 mi USD.
We also provide the number of transactions that
can be re-identified with certainty, and these
transaction values in USD.
In our interpretation, it is clear that these are violations of fiscal secrecy.

\begin{table}[!htbp]
\centering
\begin{subtable}[b]{\linewidth}
	\centering
	\begin{tabular}{lC{0.30\linewidth}C{0.2\linewidth}}
		\hiderowcolors
		\textbf{City}&\textbf{Transactions re-identified}&\textbf{Value (USD)}
		\\ \midrule
		\showrowcolors
		OURO BRANCO (MG)& 28 		&  24.6 mi 		\\
		IPAUSSU (SP)	& 3 		&  10 mi 	\\
		G.\ DO NORTE (MT)& 5 	&  9.9 mi	\\
		SAO GABRIEL (RS)& 66 &  7.3 mi	\\			
		COMODORO (MT)	& 3 &  4.1 mi 	\\
	\end{tabular}
	\caption{Cities with the largest values of transactions re-identified with certainty.}
	\label{tab:most-violated-cities}
\end{subtable}

\begin{subtable}[b]{\linewidth}
	\centering
	\begin{tabular}{C{0.16\linewidth}C{0.18\linewidth}C{0.14\linewidth}C{0.20\linewidth}C{0.12\linewidth}}
		\textbf{Attack Phase/Step} & \textbf{Prob.\ of re-identification} &  \textbf{Mult.\ Leakage} & \textbf{Transactions re-identified} & \textbf{Value (USD)} \\ \hline
		Prior & 0.05\% & -- & 0 & 0 \\ 
		1/1 & 0.10\% &  $16\times$ & 91 & 3 mi \\
		1/2 & 0.44\% & $80\times$  & 980 & 99 mi \\
		3/1 & 0.65\% & $120\times$ &  1,463 & 110 mi \\
		3/2 & 0.90\% & $168\times$ &  2,003 & 137 mi\\
	\end{tabular}
	\caption{Metrics of attack success by phase and step.}
	\label{tab:final-results}
\end{subtable}

    \caption{Experimental results for the concrete case study..}
    \label{fig:experimental-results}
\end{table}

\section{Related work}\label{sec:related-work}

As has been previously demonstrated in the Netflix prize \citep{narayanan2006break} and AOL\footnote{https://www.nytimes.com/2006/08/09/technology/09aol.html} data leaks, techniques such as the ones used by the Brazilian Government are weak against database reconstruction attacks. Even when de-identification, pseudonymization and generalization are used, the data release usually enables the construction of a set of constraints which are sufficiently restrictive as to agree with just one arrangement of the microdata.~\citep{10.1145/3287287}
	
An Australian report~\cite{HealthDataAustralia} shows how database reconstruction attacks can be used to re-identify individuals in the Australian health records public datasets. These Australian datasets contain billing information regarding procedures undertaken by individuals and paid by the government and their insurers. These billing records are de-identified and pseudonymized, nonetheless, the authors argue that around 900,000 individuals have unique sums regarding their paid expenditures and that the fact that they are unique can aid in the re-identification. Of course, those expenditure sums were calculated because the records were pseudonymized, \textit{i. e.} there was a field used for grouping the transactions. We think that the algorithm presented here shows that even if the artificially created ids are stripped from the public datasets the same sums could be reconstructed.
	
In the context of Brazil, a recent study has shown how 
the national educational census database, released by Brazilian authorities, can be used to reveal if a student has a physical disability with 99.69\% probability of success~\cite{GabrielNunes}.
The procedure we present here shows that even if the Brazilian authorities started to only release certain sensitive attributes in summaries, i. e., just disclose the total number of students with disabilities per school, re-identification would be possible, although the chance of success could be lower.


\section{Conclusion and prospects}\label{sec:conclusion}

In this paper we have presented a re-identification attack against a Brazilian dataset of foreign-trade statistics which successfully identified importers for transactions totalling 137 mi USD, thus violating those importers' rights to fiscal secrecy.
The main real-world practical consequence of our work was to inform the Brazilian Tax and Customs Administration (RFB) of the vulnerabilities we observed, in June 17th, 2021. Later in 2021, RFB removed access to the dataset named \tabimporttrans (Ordinance 100 of December \nth{16} of 2021).
We cannot state that this removal was a direct consequence of our work, but it is sufficient to avoid the attack presented here. Also, the periodicity of the \tabimporters dataset changed in February \nth{7}, 2022 and a technical note\footnote{\url{https://balanca.economia.gov.br/balanca/metodologia/Nota_publicacao_anual_da_lista_de_empresas_exportadoras_e_importadoras.pdf}} stated clearly that this was due to \qm{the possible vulnerabilities to the fiscal secrecy regarding linking this dataset with other open data}.

As mentioned before, since Brazilian foreign-trade statics follow international standards, it is possible that attacks similar to ours
may still be of concern to other countries.

One immediate goal created by our work is the proposal of techniques within the realm of differential privacy and its variants to allow for the publication of data maintaining statistical utility while enforcing privacy. A primary goal should be to introduce enough uncertainty so that all importers are able to plausibly deny being responsible for any given transaction, therefore the total amount of published data cannot contain enough information so as to violate the privacy of any importer with full certainty. In the case of foreign-trade statistics, ``plausible deniability'' is a sufficient guarantee to deter malicious actors from making reliable inferences; the requirements for utility are yet to be fully investigated. We leave this goal as future work.

One important aspect we overlooked is the fact that any real-world attacker will make use of economic common-sense --- a mining company is unlikely to import car parts, for example. Potential work in this area would be to incorporate machine learning methods for pre-processing data to provide a stronger (and perhaps more realistic) prior for the adversary against which new privacy defences could be evaluated.

In terms of the optimization modelling and attack algorithm, several research prospects unfold from this work. First, it would be interesting to implement database reconstruction attacks with tailored implementation of backtracking methods or dynamic programming, as these should be more scalable and less prone to numerical instability. The necessity to deal with errors might lead to imperatively complicated implementations. Second, the framework of integer linear programs more often than not leads to efficiently solvable approximation algorithms, which we intend to investigate both in the theoretical model of \texttt{Package Allocation} and also in the practical setting of our work. Third, and finally, the  introduction of noise to the data deforms the feasibility region of the program. The inverse optimization problem of determining which direction in the feasible set is least affected by the noise can lead to very interesting considerations. Assuming noise shall be eventually introduced to data with the intent of maximizing privacy protection and minimizing loss of utility, it might be possible to frame this task as a geometric modelling problem. We leave this investigation to future work.

\bibliographystyle{plain}
\bibliography{references}
\vspace{2cm}
\begin{center}
Author affiliations
\end{center}

\noindent    \textsc{Danilo Fabrino Favato}\\
    \textsc{Dept. of Computer Science} \\ 
    \textsc{Universidade Federal de Minas Gerais, Brazil} \\
    \textit{E-mail address}: \texttt{dfavato@dcc.ufmg.br} \\ \ \\
    \textsc{Gabriel Coutinho} \\
    \textsc{Dept. of Computer Science} \\ 
    \textsc{Universidade Federal de Minas Gerais, Brazil} \\
    \textit{E-mail address}: \texttt{gabriel@dcc.ufmg.br}\\ \ \\
    \textsc{M\'{a}rio S.\ Alvim} \\
    \textsc{Dept. of Computer Science} \\ 
    \textsc{Universidade Federal de Minas Gerais, Brazil} \\
    \textit{E-mail address}: \texttt{msalvim@dcc.ufmg.br}\\ \ \\
    \textsc{Natasha Fernandes} \\
    \textsc{Macquarie University, Australia} \\
    \textit{E-mail address}: \texttt{tashfernandes@gmail.com}

\newpage
\appendix
\section{Detailed description of the published datasets}\label{sec:sargasso-dataset-structure}

Tbls.~\ref{tab:ncm-uf-fields-complete}, 
\ref{tab:rfb-fields-complete},
\ref{tab:importers-fields-complete}, and
\ref{tab:hs4-city-fields-complete}
below provide the detailed structure of of the published
datasets used in our attacks, they were partially described in \Sec{sec:overview-datasets}.

\begin{table*}[!htbp]
	\centering
    \begin{tabular}{L{0.10\linewidth}L{0.50\linewidth}L{0.30\linewidth}}
		\textbf{Field} & \textbf{Description} & \textbf{Example} \\\hline
		\texttt{CO\_ANO}			& Year  when the import
		was registered.		& 2020				\\
		\texttt{CO\_MES}			& Month when the import
		was registered.		& 10				\\
		\texttt{CO\_NCM}			& Imported goods NCM
		code number.		& 40161090 (Oth.works of 
		vulcan. rubber 
		alveol. n/harden.)\\
		\texttt{CO\_PAIS}		& Imported goods country
		of origin.			& 249 (USA)			\\
		\texttt{SG\_UF\_NCM}	& Importer federation state. & MG  \\
		\texttt{CO\_VIA}			& Imported goods
		entrance route.		& 04 (AERIAL)		\\
		\texttt{CO\_URF}			& Unit of the \rfb
		responsible for
		administrative 
		treatment of the
		imported good.		& 0617700 
		(BELO HORIZONTE)	\\
		\texttt{CO\_UNID}		& Unit of measurement.	& 10 (kg)			\\
		\texttt{QT\_ESTAT}		& Imported goods total
		quantity in the 
		given unit 
		(CO\_UNID).			& 53				\\
		\texttt{KG\_LIQUIDO}		& Imported goods total net weight in kg.	& 53	 \\
		\texttt{VL\_FOB}			& Imported goods total value in USD. & 3,236	\\
		\texttt{VL\_FRETE}	& Total international freight in USD & 20 \\			
		\texttt{VL\_SEGURO}	& Total international insurance in USD & 2 \\		
	\end{tabular}
	\caption{Complete structure of the \tabsumbyncm dataset.}
	\label{tab:ncm-uf-fields-complete} 
\end{table*}

\begin{table*}[!htbp]
	\centering
\footnotesize{		
\begin{tabular}{L{0.15\linewidth}L{0.40\linewidth}L{0.40\linewidth}}
			\textbf{Field} & \textbf{Description} & \textbf{Example} \\\hline
			\texttt{NUMERO DE ORDEM} & Sequential unique number aiding the identification of the Declaration of Imports (DI). & 150320000100001 \\ 
            \texttt{ANOMES} & Year and month of transaction registry. & 202011 \\ 
            \texttt{COD.NCM} & DI header: NCM code number of imported good. & 42010090 (Saddlery and harness of other materials) \\ 
            \texttt{PAIS DE ORIGEM} & DI header: country of origin. & CHINA \\ 
            \texttt{PAIS DE AQUISICAO} & DI header: country of acquisition. & CHINA \\ 
            \texttt{UNIDADE DE MEDIDA} & DI header: measurement unit for statistical purposes. & NET KILOGRAM \\ 
            \texttt{UNIDADE COMERC.} & DI item: measurement unit for commercial purposes. & PIECE \\ 
            \texttt{DESCRICAO DO PRODUTO} & DI item: detailed description of imported good, as submitted by the importer (usually in Portuguese). & Bandana para animal de estimação, composição 100\% algodão, sem forro. MARCA: ACCESSORI - Artigo: 13556843 (P/N: 9.IMPOR.182.000003) (P/N Fab.: BANDANA ESTRELA)
            \\ 
            \texttt{QTDE ESTATISTICA} & DI header: statistical quantity. & 477.3 \\ 
            \texttt{PESO LIQUIDO} & DI header: net weight (in kg). & 477.3 \\ 
            \texttt{VMLE DOLAR} & DI header: FOB value in USD. & 8,313.44 \\ 
            \texttt{VL FRETE DOLAR} & DI header: international freight in USD. & 245.23 \\ 
            \texttt{VL SEGURO DOLAR} & DI header: insurance in USD. & 0.00 \\ 
            \texttt{VALOR UN. PROD.DOLAR} & Ratio \texttt{TOT.UN.PROD.DOLAR} / \texttt{QTD COMERCIAL} & 0.87 \\ 
            \texttt{QTD COMERCIAL} & DI item: commercialized quantity. & 4,773 \\ 
            \texttt{TOT.UN. PROD.DOLAR} & DI item: value in USD. & 4,152.51 \\ 
            \texttt{UNIDADE DESEMBARQUE} & DI header: port of entrance. & NOT/INFORMED \\ 
            \texttt{UNIDADE DESEMBARACO} & \rfb unit administratively treating this DI. & PORTO DE SANTOS \\ 
            \texttt{INCOTERM} & International commercial terms of the DI. & FOB \\ 
            \texttt{NAT. INFORMACAO} & DI header: nature of the operation. & EFFECTIVE \\	
        \end{tabular}
		\caption{Complete structure of the \dataset{De-identified transactions} dataset.}
		\label{tab:rfb-fields-complete} 
}	
\end{table*}

\begin{table*}[!htbp]
	    \centering
		\begin{tabular}{L{0.10\linewidth}L{0.30\linewidth}L{0.50\linewidth}}
			\textbf{Field} & \textbf{Description} & \textbf{Example} \\\hline
			\texttt{CNPJ}		& Business registration number
			in Brazil.					& 165.823.950/001-18	\\
			\texttt{EMPRESA}		& Business name.				& A. A. TELECHESKI		\\
			\texttt{ENDEREÇO}	& Business address.				& RUA INGLATERRA		\\
			\texttt{NÚMERO}		& Business address number.		& 160					\\
			\texttt{BAIRRO}		& Business address neighborhood & JARDIM EUROPA			\\
			\texttt{CEP}			& Business zip code.			& 87111-090				\\
			\texttt{MUNICÍPIO}	& Business address city.		& SARANDI				\\
			\texttt{UF}			& Business federation state.	& PR					\\
			\texttt{CNAE} 
			\texttt{PRIMÁRIA}	& Economic nature / industry
			segment of the business		& 4649 - 
			Wholesale commerce
			of equipment and
			personal articles
			for personal and
			domestic usage not
			previously specified.		\\
			\texttt{NATUREZA
			JURÍDICA}	& Juridic nature of the business& 213 - Individual 
			merchant firm	 \\
		\end{tabular}
		\caption{Complete structure of the \dataset{Importers} dataset.}
		\label{tab:importers-fields-complete}
	\end{table*}	

	\begin{table*}[!htbp]
	    \centering
		\begin{tabular}{L{0.10\linewidth}L{0.30\linewidth}L{0.50\linewidth}}
			\textbf{Field} & \textbf{Description} & \textbf{Example} \\\hline
			\texttt{CO\_ANO}			& Year  when the import
			was registered.		& 2021								\\
			\texttt{CO\_MES}			& Month when the import
			was registered.		& 01								\\
			\texttt{SH4}				& Imported goods SH4 
			code number.		& 3002 (Human blood; animal blood 
			prepared for therapeutic, 
			prophylactic or diagnostic uses; 
			antisera and other blood 
			fractions and modified 
			immunological products, 
			whether or not obtained by 
			means of biotechnological 
			processes; vaccines, toxins, 
			cultures of micr)				\\
			\texttt{CO\_PAIS}		& Imported goods country
			of origin.			& 249 (USA)							\\
			\texttt{SG\_UF\_MUN}		& Importer federation
			state.				& SP								\\
			\texttt{CO\_MUN}			& Importer city.		& 3449904 (SAO JOSE DOS CAMPOS)			\\
			\texttt{KG\_LIQUIDO}		& Imported goods total
			net weight in kg.	& 326								\\
			\texttt{VL\_FOB}			& Imported goods total
			value in USD.		& 906,723	\\			
		\end{tabular}
		\caption{Complete structure of the \dataset{Summary by city} dataset.}
		\label{tab:hs4-city-fields-complete} 
			
	\end{table*}
	

\end{document}